\documentclass{llncs}

\usepackage{stmaryrd}

\usepackage{amssymb, amsmath} 
\usepackage{graphicx}
\usepackage[pdftex,dvipsnames]{xcolor}  
\usepackage[T1]{fontenc} 
\usepackage{hyperref} 

\usepackage{listings}

\definecolor{commentgreen}{RGB}{2,112,10}
\definecolor{eminence}{RGB}{108,48,130}
\definecolor{weborange}{RGB}{200,130,0}
\definecolor{frenchplum}{RGB}{129,20,83}
\usepackage{textcomp}

\newcommand{\prettylstciao}[0]{
\lstset{language=Prolog,
        frame=ltrb,
        rulecolor=\color{blue},
        tabsize=4,
        showstringspaces=false,
        breaklines=true,breakatwhitespace=true,
        showlines=true,
        showspaces=false,showtabs=false,
        commentstyle=\color{gray},
    keywordstyle=\color{eminence},
    stringstyle=\color{red},
    basicstyle=\scriptsize\ttfamily, 
    keywordstyle=\color{weborange},
    emphstyle={\color{blue}},
    emph={pred,prop,trust,check,checked,true,rsize,cardinality,not_fails,module,exp,cost,costb,
      steps_ub,steps_lb,size_ub,size_lb,covered,mut_exclusive,cost,use_module,int,calls,success,head_cost,literal_cost,
      mshare,trust_default,int,atm,term,comp,
      is_det,num,var,list,ground,length,terminates,steps_o,resource,entry,impl_defined},
    otherkeywords={>,<,>=,=<,.,;,-,!,=,~,*,\&,+,:-,[,],|,->,:,:=},
    morekeywords= {>,<,>=,=<,.,;,-,!,=,~,*,\&,+,:-,[,],|,->,:,:=},
    escapechar=@,
    escapeinside=~~,
      }}

\graphicspath{{graphics/}{Figs/}{graphics/subdir/}}
\usepackage{galois} 

\newenvironment{example-box}{\begin{example}\rm}{$\Box$\end{example}}

\usepackage[utf8]{inputenc}

\usepackage{lineno}
\modulolinenumbers[5]
\newcommand{\powerset}[1]{\wp{#1}}

\reversemarginpar

\newcommand{\secpre}{\vspace*{-2mm}}
\newcommand{\secpost}{}
\newcommand{\subsecpre}{\vspace*{-3mm}}
\newcommand{\subsecpost}{\vspace*{-1mm}}

\begin{document}

\title{\vspace*{-3mm}Computing Abstract Distances in Logic Programs}

\author{
       Ignacio Casso\inst{1,2} \and
       José F. Morales\inst{1} \and \\ 
       Pedro Lopez-Garcia\inst{1,3} \and
       Manuel V. Hermenegildo\inst{1,2}
       }

\institute{IMDEA Software Institute \and
          T. University of Madrid (UPM) \and
          Spanish Council for Scientific Research (CSIC)
          }

\maketitle

\begin{abstract}

  Abstract interpretation is a well-established technique for
  performing static analyses of logic programs. However, choosing the
  abstract domain, widening, fixpoint, etc.\ that provides the best
  precision-cost trade-off remains an open problem. This is in a good
  part because of the challenges involved in measuring and comparing
  the precision of different analyses.
  We propose a new approach for measuring such precision, based on
  defining distances in abstract domains and extending them to
  distances between whole analyses of a given program, thus allowing
  comparing precision across different analyses.
  We survey and extend existing proposals for distances and metrics in
  lattices or abstract domains, and we propose metrics for some 
  common domains used in logic program analysis, as well as extensions
  of those metrics to the space of whole program analysis. We
  implement those metrics within the CiaoPP framework and apply them
  to measure the precision of different analyses over
  both benchmarks and a realistic program.

\end{abstract}



\secpre
\section{Introduction}
\secpost

Many practical static analyzers for (Constraint) Logic Programming
((C)LP) are based on the theory of Abstract
Interpretation~\cite{Cousot77-short}.
The basic idea behind this technique is to interpret (i.e., execute)
the program over a special abstract domain
to obtain some abstract semantics
of the program,
which will over-approximate every possible execution
in the standard (concrete)
domain.
This makes it possible to reason safely (but perhaps imprecisely)
about the properties that hold for all such executions.
As mentioned before, abstract interpretation has proved practical and
effective for building static analysis tools, and in particular in the
context of (C)LP~\cite{abs-int-naclp89-short,pracai-jlp-short,van-roy-computer,effofai,deb-acm89,bruy91,Ma-So-Jo,anconsall-acm-shorter,clpropt-short}.

Recently, these techniques have also been applied successfully to the
analysis and verification of other programming paradigms by using
(C)LP (Horn Clauses) as the intermediate representation for
different compilation levels, ranging from source to bytecode or ISA
~\cite{jvm-cost-esop-short-plus,BandaG08-short,resources-bytecode09-shorter,DBLP:conf/tacas/GrebenshchikovGLPR12-short,isa-energy-lopstr13-final-shortest,DBLP:conf/tacas/AngelisFPP14-shorter,DBLP:conf/cav/GurfinkelKKN15-shorter,DBLP:conf/birthday/BjornerGMR15-shorter,DBLP:conf/pldi/MadsenYL16-short,kafle-cav2016-short}.

When designing or choosing an abstract interpretation-based analysis,
a crucial issue is the trade-off between cost and precision, and thus
research in new abstract domains, widenings, fixpoints, etc., often
requires studying this
trade-off.
However, while measuring analysis cost is typically relatively
straightforward, having effective precision measures is much more
involved.
There have been a few proposals 
for this purpose, including, e.g., probabilistic abstract
interpretation~\cite{pierro-01} and some measures in numeric
domains~\cite{logozzo-09,sotin-10}
\footnote{Some of these attempts (and others) are further explained in
  the related work section (Section \ref{related-work}).}
, but they have limitations and in practice most studies
come up with ad-hoc measures for measuring 
precision.
Furthermore, there have been no proposals for such measures in (C)LP
domains.

We propose a new approach for measuring the precision of abstract
interpretation-based analyses in (C)LP, based on defining
\emph{distances in abstract domains} and extending them to
\emph{distances between whole analyses of a given program}, which
allow comparison of precision across different analyses.  Our
contributions can be summarized as follows:
We survey and extend existing proposals for distances in lattices and
abstract domains (Sec.~\ref{domain-distances}).
We then build on this theory and ideas to propose distances for common
domains used in (C)LP analysis~(Sec. \ref{clp-distances}).
We also propose a principled methodology for comparing quantitatively
the precision of different abstract interpretation-based analyses of a
whole program (Sec.~\ref{analysis-distances}).
This methodology is parametric on the distance in the
underlying abstract domain and only relies in a unified representation
of those analysis results as AND-OR trees. 
Thus, it can be used
to measure the precision of new fixpoints, widenings, etc. within a
given abstract interpretation framework, not requiring knowledge of
its implementation.
To the extent of our knowledge, all previous principled attempts at
measuring the precision of different abstract interpretations have
addressed the precision of analysis operators, rather than providing a
general methodology for comparing the results obtained for particular
programs.
Finally, we also provide experimental evidence about the
appropriateness of the proposed distances (Sec.~\ref{experiments}).
\secpre
\section{Background and Notation}
\secpost

\paragraph{Lattices:}  A \emph{partial order} on a set \(X\)
is a binary relation \(\sqsubseteq\) that is reflexive, transitive,
and antisymmetric. The \emph{greatest lower bound} or \emph{meet} of
$a$ and $b$, denoted by $a \sqcap b$, is the greatest element in \(X\)
that is still lower than both of them ($a \sqcap b \sqsubseteq a, ~a
\sqcap b \sqsubseteq b, ~(c \sqsubseteq a\ \land c \sqsubseteq b
\implies c \sqsubseteq a \sqcap b))$. If it exists, it is unique. The
\emph{least upper bound} or \emph{join} of $a$ and $b$, denoted by $a
\sqcup b$, is the smallest element in $X$ that is still greater than
both of them ($a \sqsubseteq a \sqcup b, ~ b \sqsubseteq a \sqcup b, ~
(a \sqsubseteq c \land b \sqsubseteq c \implies a \sqcup b \sqsubseteq
c)$). If it exists, it is unique.  A partially ordered set (poset) is
a couple $(X, \sqsubseteq)$ such that the first element $X$ is a set
and the second one is a partial order relation on $X$. A
\emph{lattice} is a poset for which any two elements have a meet and a
join. A lattice $L$ is complete if, extending in the natural way the
definition of supremum and infimum to subsets of $L$, every subset $S$
of $L$ has both a supremum $sup(S)$ and an infimum $inf(S)$. The
maximum element of a complete lattice, $\sup(L)$ is called \textit{top} or
$\top$, and the minimum, $\inf(L)$ is called \textit{bottom} or $\bot$.

\vspace*{-3mm}
\paragraph{Galois Connections:} Let \((L_1, \sqsubseteq_1)\) and \((L_2, \sqsubseteq_2)\) be two
posets. Let \(f: L_1 \longrightarrow L_2\) and \(g: L_2
\longrightarrow L_1\) be two applications such that:
\vspace*{-1mm}
\[\forall x \in L_1, y \in L_2: f(x)\,\sqsubseteq_2\, y \iff x\,
\sqsubseteq_1\, g(y)\]\\ [-4mm]
\vspace*{-1mm}
\noindent
Then the quadruple \(\langle L_1,f,L_2,g \rangle\) is a \emph{Galois
  connection}, written $L_1\galois{f}{g} L_2$. If
$f \circ g$ is the identity, then the quadruple is called a
\emph{Galois insertion}.

\vspace*{-2mm}
\paragraph{Abstract Interpretation and Abstract Domains:}
Abstract interpretation~\cite{Cousot77-short} is a well-known static
analysis technique that 
allows computing sound over-approx\-imations of the semantics of
programs. The semantics of a program can be described in terms of the
\textit{concrete domain}, whose values in the case of (C)LP
are typically sets of variable 
substitutions that may occur at runtime. The idea behind abstract
interpretation is to interpret the program over a special abstract
domain, whose values, called \textit{abstract substitutions}, are
finite representations of possibly infinite sets of actual
substitutions in the concrete domain. We will denote the concrete
domain as $D$, and the abstract domain as $D_\alpha$. We will denote the
functions that relate sets of concrete substitutions with abstract
substitutions as the \textit{abstraction} function $\alpha: D
\longrightarrow D_\alpha$ and the \textit{concretization} function
$\gamma: D_\alpha \longrightarrow D$.
The concrete domain is a complete lattice under the set inclusion
order, and that order induces an ordering relation in the abstract
domain herein represented by ``$\sqsubseteq$.'' Under this relation
the abstract domain is usually a complete lattice or cpo
and $(D,\alpha,D_\alpha,\gamma)$ is a Galois insertion. The abstract
domain is of finite height or alternatively it is equipped with a
\textit{widening operator}, which allows for skipping over infinite
ascending chains during analysis to a greater fixpoint, achieving
convergence in exchange for precision.

\vspace*{-3mm}
\paragraph{Metric:} A metric on a set $S$ is a function
$d : S \times S \rightarrow \mathbb{R}$ satisfying:
\vspace*{-2mm}
\begin{itemize}
\item Non-negativity: \hfill $\forall x,y \in S,~ d(x,y) \geq 0$.

\item Identity of indiscernibles: \hfill $\forall x,y \in S,~ d(x,y)=0 \iff x=y$.

\item Symmetry: \hfill
  $\forall x,y \in S,~ d(x,y) = d(y,x)$.

\item Triangle inequality: \hfill
  $\forall x,y,z \in S,~ d(x,z) \leq
  d(x,y)+d(y,z)$.

\end{itemize}

A set $S$ in which a metric is defined is called a metric space.  A
pseudometric is a metric where two elements which are different are
allowed to have distance 0. We call the left implication of the
identity of indiscernibles, weak identity of indiscernibles.
A well-known method to extend a metric
$d: S \times S \longrightarrow \mathbb{R}$ to a metric in
$\powerset(S)$ is using the Hausdorff distance, defined as:
\vspace*{-3mm}
\[\text{d}_{H}(A,B) = \max\left\{ \sup_{a\in A} \inf_{b\in
  B} \text{d}(a,b),\sup_{b\in B} \inf_{a\in A}\text{d}(a,b)\right\}\]

\secpre
\section{Distances in Abstract Domains}
\secpost
\label{domain-distances}

As anticipated in the introduction, our distances between abstract
interpretation-based analyses of a program will be parameterized by
distance in the underlying abstract domain, which we assume to be a
complete lattice. In this section we propose a few such distances for
relevant logic programming abstract domains. But first we review and
extend some of the concepts that
arise when working with lattices or abstract domains as metric spaces.

\subsecpre
\subsection{Distances in lattices and abstract domains}
\subsecpost

When defining a distance in a partially ordered set, it is necessary
to consider the compatibility between the metric and the structure of
the lattice. This relationship will suggest new properties that a
metric in a lattice should satisfy. For example, a distance in a
lattice should be \textit{order-preserving}, that is,
$\forall a,b,c \in D ~ with ~ a \sqsubseteq b \sqsubseteq c, ~ then ~
d(a,b),d(b,c) \leq d(a,c)$. It is also reasonable to expect that it
fulfills what we have called the diamond inequality, that is,
$\forall a,b,c,d \in D ~ with ~ c \sqcap d \sqsubset a \sqcap b, ~ a
\sqcup b \sqsubset c \sqcup d, ~ then ~ d(a,b) \leq d(c,d)$. But more
importantly, this relationship will suggest insights for constructing
such metrics.

One such insight is precisely defining a partial metric
$d_\sqsubseteq$ only between elements which are related in the
lattice, which is arguably easier, and to extend it later to a
distance between arbitrary elements $x,y$, as a function of
$d_\sqsubseteq(x,x \sqcap y),~d_\sqsubseteq(y,x \sqcap
y),~d_\sqsubseteq(x,x \sqcup y), ~d_\sqsubseteq(x,x \sqcup y)$ and
$d_\sqsubseteq(x \sqcap y, x \sqcup y)$. Jan Ramon et
al.~\cite{DBLP:conf/ilp/RamonB98-short} show under which circumstances 
$d_\sqsubseteq(x,x \sqcup y) + d_\sqsubseteq(y, x \sqcup y)$ is a
distance, that is, when $d_\sqsubseteq$ is order-preserving and
fulfills
$d_\sqsubseteq(x,x \sqcup y) + d_\sqsubseteq(y, x \sqcup y) \leq
d_\sqsubseteq(x,x \sqcap y) + d_\sqsubseteq(y, x \sqcap y)$.

In particular, one could define a monotonic size
$size: L \rightarrow \mathbb{R}$ in the lattice and define
$d_\sqsubseteq(a,b)$ as $size(b)-size(a)$. 
Gratzer \cite{general-lattice-theory} shows
that if the size
fullfills $size(x) + size(y) = size(x \sqcap y) + size(x \sqcup y)$,
then $d(x,y)=size(x \sqcup y) - size(x \sqcap y)$ is a metric. De
Raedt~\cite{raedt-09} shows that
$d(x,y) = size(x) + size(y) - 2 \cdot size(x \sqcup y)$ is a metric iff
$size(x)+size(y) \leq size(x \sqcap y) + size(x \sqcup y)$, and an
analogous result with $d(x,y) = size(x) + size(y) - 2 \cdot size(x \sqcup y)$
and $\geq$ instead of $\leq$. Note that the first distance is the
equivalent of the \textit{symmetric difference distance} in finite
sets, with $\sqsubseteq$ instead of $\subseteq$ and $size$ instead of
the cardinal of a set. Similar distances for finite sets, such as the
Jaccard distance, can be translated to lattices in the same way.
Another approach to defining $d_\sqsubseteq$ that follows from the
idea of using the lattice structure, is counting the steps between two
elements (i.e., the number of edges between both elements in the Hasse
diagram of the lattice). This was used by~Logozzo~\cite{logozzo-09}.

When defining a distance not just in any lattice, but in an actual abstract
domain (\textit{abstract distance from now on}), it is also necessary
to consider the relation of the abstract domain with the concrete
domain (i.e., the Galois connection), and how an abstract distance is
interpreted under that 
relation. In that sense, we can observe that a distance $d_{D_\alpha}:
D_\alpha \rightarrow D_\alpha$ in an abstract domain will induce a
distance $d_D^\alpha: D \rightarrow D$ in the concrete one, as
$d_D^\alpha(A,B)=d_{D_\alpha}(\alpha(A),\alpha(B))$, and the other way
around: a distance $d_D : D \rightarrow D$ in the concrete domain
induces an abstract distance $d_{D_\alpha}^\gamma: D_\alpha
\rightarrow D_\alpha$ in the abstract one, as
$d_{D_\alpha}^\gamma(a,b)=d_D(\gamma(a),\gamma(b))$. Thus, an abstract
distance can be interpreted as an abstraction of a distance in the
concrete domain, or as a way to define a distance in it, and it is
clear that it is when interpreted that way that an abstract distance
makes most sense from a program semantics point of view.

It is straightforward to see
(and we show in the appendix)
that these induced distances inherit most metric and order-related
properties. In particular, if a distance $d_D$ in the concrete domain
is a metric, its abstraction $d_{D_\alpha}$ is a pseudo-metric in the
abstract domain, and a full metric if the Galois connection between
$D$ and $D_\alpha$ is a Galois insertion. This allows us to define
distances $d_\alpha$ in the abstract domain from distances $d$ the
concrete domain, as $d_\alpha(a,b)=d(\gamma(a),\gamma(b))$.  This
approach might seem of little applicability, due to the fact that
concretizations will most likely be infinite and we still need metrics
in the concrete domain. But in the case of logic programs, such
metrics for Herbrand terms already exist (e.g.,
\cite{hutch-ecml-97,cheng-hwei-ilp-97,DBLP:conf/ilp/RamonB98-short}),
and in fact we show later a distance for the \textit{regular types}
domain that can be interpreted as an extension of this kind, of the
distance proposed by Nienhuys-Cheng~\cite{cheng-hwei-ilp-97} for sets
of terms.

Finally, we note that a metric in the Cartesian product of lattices
can be easily derived from existing distances in each lattice, for
example as the 2-norm or any other norm of the vector of distances
component to component. This is relevant because many abstract
domains, such as those that are combinations of two different abstract
domains, or non-relational domains which provide an abstract value
from a lattice for each variable in the substitution, are of such
form. However, although this is a well-known result, it is not clear
whether the resulting distance will fulfill other lattice-related
properties if the distances for each component do. It is straightforward to see that 
that is the case for the \textit{order-preserving} property, but not
for the \textit{diamond inequality}, due to the fact that for abstract
domains, all elements of the lattice $(a_1, \ldots, a_n)$ for which
$\exists i ~ s.t. ~ a_i=\bot$ are identified as the bottom element of
the cartesian product lattice, since their concretization is
$\emptyset$.

\subsecpre
\subsection{Distances in Logic Programming Domains}
\subsecpost
\label{clp-distances}

We now
propose some
distances for two well-known abstract domains used in (C)LP,
following the considerations presented in the previous section.

\vspace*{-2mm}
\paragraph{Sharing domain:}

The \texttt{sharing} domain~\cite{jacobs89-short,abs-int-naclp89-short} is a
well-known domain for analyzing the sharing (aliasing) relationships
between variables and grounding in logic programs. It is defined as
$\wp(\wp(Pvar))$, that is, an abstract substitution for a clause is
defined to be {\em a set of sets of program variables} in that clause,
where each set indicates that the terms to which those variables are
instantiated at runtime might share a free variable. More formally, we
define $Occ(\theta,U) = \{X|X \in dom(\theta), U \in vars(X\theta)\}$,
the set of all program variables $X \in Pvar$ in the clause such that
the variable $U \in Uvar$ appears in $X\theta$. We define the
abstraction of a substitution $\theta$ as ${\cal A}_{sharing}(\theta)
~ = ~ \{Occ(\theta,U) \;|\; U \in Uvar\}$, and extend it to sets of
substitutions. The order induced by this abstraction in
$\wp(\wp(Pvar))$
is the set inclusion, the join, the set union, and the meet, the set
intersection. As an example, a program variable that does not appear
in any set is guaranteed to be ground, two variables that never appear
in the same set are guaranteed to not share, or
$\top=\wp(Pvar)$.  The complete definition can be found
in~\cite{jacobs89-short,abs-int-naclp89-short}).

Following the approach of previous section, we define this monotone
size in the domain: $size(a) = |a|+1, size(\bot)=0$. It is
straightforward to check that
$\forall a,b \in Sh, ~ size(a) + size(b) = size(a \sqcap b) + size(a
\sqcup b)$. Therefore the following distance is a metric and
order-preserving:
\vspace*{-3mm}  
\[d_{share}(Sh_1,Sh_2) = size(Sh_1 \cup Sh_2) - size(Sh_1 \cap Sh_2) =
  |(Sh_1 \cup Sh_2)| - |size(Sh_1 \cap Sh_2)|\]
\noindent
We would like our distance to be in a normalized range $[0,1]$, and
for that we divide it between $d(\bot,\top)=2^n$, where $n=|V|$
denotes the number of variables in the domain of the
substitutions. This yields the following final distance, which is a
metric by construction: \vspace*{-2mm}
\[d_{share}(Sh_1,Sh_2) = (|(Sh_1 \cup Sh_2)| - |size(Sh_1 \cap Sh_2)|)/2^n\]

\vspace*{-5mm}  
\paragraph{Regular-type domain:}

Another well-known domain for logic programs is the \textit{regular
  types} domain~\cite{Dart-Zobel}, which abstracts the shape or type of the
terms to which variables are assigned on runtime. It associates each
variable with a deterministic context free grammar that describes its
shape, with the possible functors and atoms of the program as terminal
symbols. A more formal definition can be found in~\cite{Dart-Zobel}.
We will write abstract substitutions as tuples
$\langle T_1,\ldots,T_n \rangle$, where
$T_i = (S_i,{\cal T}_i,{\cal F}_i,{\cal R}_i)$ is the grammar that
describes the term associated to the i-th variable in the
substitution.
We propose to use as a basis 
the Hausdorff distance in the concrete domain, using the distance
between terms proposed in \cite{cheng-hwei-ilp-97}, i.e.,

$\\
\begin{array}{l}
d_{term}(f(x_1,\ldots,x_n),g(y_1,\ldots,y_m)) =
 \left\{ \begin{array}{ll}
        if & f/n \neq g/m 
        \;\;\;\; then \;\; 1 \\
        else & p\sum_{i=1}^n{\frac{1}{n}d_{term}(x_i,y_i)}
        \end{array} \right. 
\end{array}
\\
$ 

As the derived abstract version, we propose the following distance
between two types or grammars $S_1,~S_2$, defined recursively and with
a little abuse of notation: \\

$
\begin{array}{l}
d'(S_1,S_2) =
 \left\{ \begin{array}{ll}
        if & \exists ~(S_1 \rightarrow f(T_1,\ldots,T_n)) \in {\cal R}_1 \land \nexists (S_2 \rightarrow f(T'_1,\ldots,T'_n)) \in {\cal R}_2
        \;\;\;\; then \;\; 1 \\
        if & \exists ~(S_2 \rightarrow f(T_1,\ldots,T_n)) \in {\cal R}_2 \land \nexists (S_1 \rightarrow f(T'_1,\ldots,T'_n)) \in {\cal R}_1
        \;\;\;\; then \;\; 1 \\
        else & max\{p\sum_{i=1}^n{\frac{1}{n}d'(T_i,T'_i)} ~|~ (S_1 \rightarrow f(T_1,\ldots,T_n)) \in {\cal R}_1 \land \\
           & ~~~~~~~~~~~~~~~~~~~~~~~~~~~~~~~~~~~~~~~~~~~~~~~~~~~~~(S_2 \rightarrow f(T'_1,\ldots,T'_n)) \in {\cal R}_2\}
        \end{array} \right. 
\end{array}
\\
$

We also extend this distance between types 
to distance between substitutions in the abstract domain as follows: \\ 
\centerline{$d(\langle T_1,\ldots,T_n \rangle, \langle T'_1,\ldots,T'_n \rangle) =
\sqrt{d'(T_1,T'_1)^2 + \ldots + d'(T_n,T'_n)^2}$}

Since $d'$ is the abstraction of the Hausdorff distance with
$d_{term}$, which it is proved to be a metric in
\cite{cheng-hwei-ilp-97}, $d'$ is a metric too, as seen in the
previous section. Therefore $d$ is also a metric, since it is its
extension to the cartesian product.

\vspace*{-2mm}
\secpre
\section{Distances between analyses}
\secpost
\label{analysis-distances}

We now attempt to extend a distance in an abstract domain to distances
between results of different abstract interpretation-based analyses of
the same program over that domain. In the following we will assume
(following most ``top-down'' analyzers for (C)LP
programs~\cite{abs-int-naclp89-short,bruy91,anconsall-acm-shorter,clpropt-short})
that the result of an analysis for a given entry (i.e., an initial
predicate \textit{P}, and an initial call pattern or abstract query
$\lambda_c$), is an AND-OR tree,
with root the OR-node $\langle P, \lambda_c, \lambda_s \rangle_\lor$,
where $\lambda_s$ is the abstract substitution computed by the
analysis for that predicate given that initial call pattern.
An AND-OR tree alternates AND-nodes, which correspond to clauses in
the program, and OR-nodes, which correspond to literals in those
clauses. An OR-node is a triplet
$\langle L, \lambda_c, \lambda_s \rangle_\lor$, with \textit{L} a call
to a predicate \textit{P} and $\lambda_c, \lambda_s$ the abstract call
and success substitutions for that goal. It has one AND-node
$\langle C_j, \beta_{entry}^j, \beta_{exit}^j \rangle_\land$ as child
for each clause $C_j$ in the definition of \textit{P}, where
$\beta_{entry}^j=\lambda_c ~ \forall j$ and
$\lambda_s = \bigsqcup \beta_{exit}^j$. An AND-node is a triplet
$\langle C, \beta_{entry}, \beta_{exit} \rangle_\land$, with
\textit{C} a clause $ Head :- L_1,...,L_n$ and with
$\beta_{entry}, \beta_{exit}$ the abstract entry and exit
substitutions for that clause. It has an OR-node
$\langle L_i, \lambda_c^i, \lambda_s^i \rangle_\lor$ for each literal
$L_i$ in the clause, where
$\beta_{entry}=\lambda_c^1, ~ \lambda_s^i=\lambda_c^{i+1}, ~
\lambda_s^n=\beta_{exit}$. This tree is the abstract counterpart of
the resolution trees that represent concrete top-down executions, and
represents a possibly infinite set of those resolution trees at
once. The tree will most likely be infinite, but can be represented as
a finite cyclic tree. We denote the children of a node $T$ as $ch(T)$.

\begin{example-box}%
  Let us consider as an example the simple quick-sort program (using
  difference lists) in Fig.%
  ~\ref{fig:quicksort-abstree}, 
  which uses an \textit{entry} assertion to specify the initial
  abstract query of the analysis~\cite{assert-lang-disciplbook}.
If we analyze it with a simple \textit{groundness} domain (with just
two values \texttt{g} and \texttt{ng}, plus $\top$ and $\bot$),
the result can be represented with the graph shown in
Fig.~\ref{fig:quicksort-abstree}.
\begin{figure}[t]
\fbox{%
\begin{minipage}{0.33\textwidth}
  \includegraphics[scale=0.5,clip,trim=39 39 40 38]{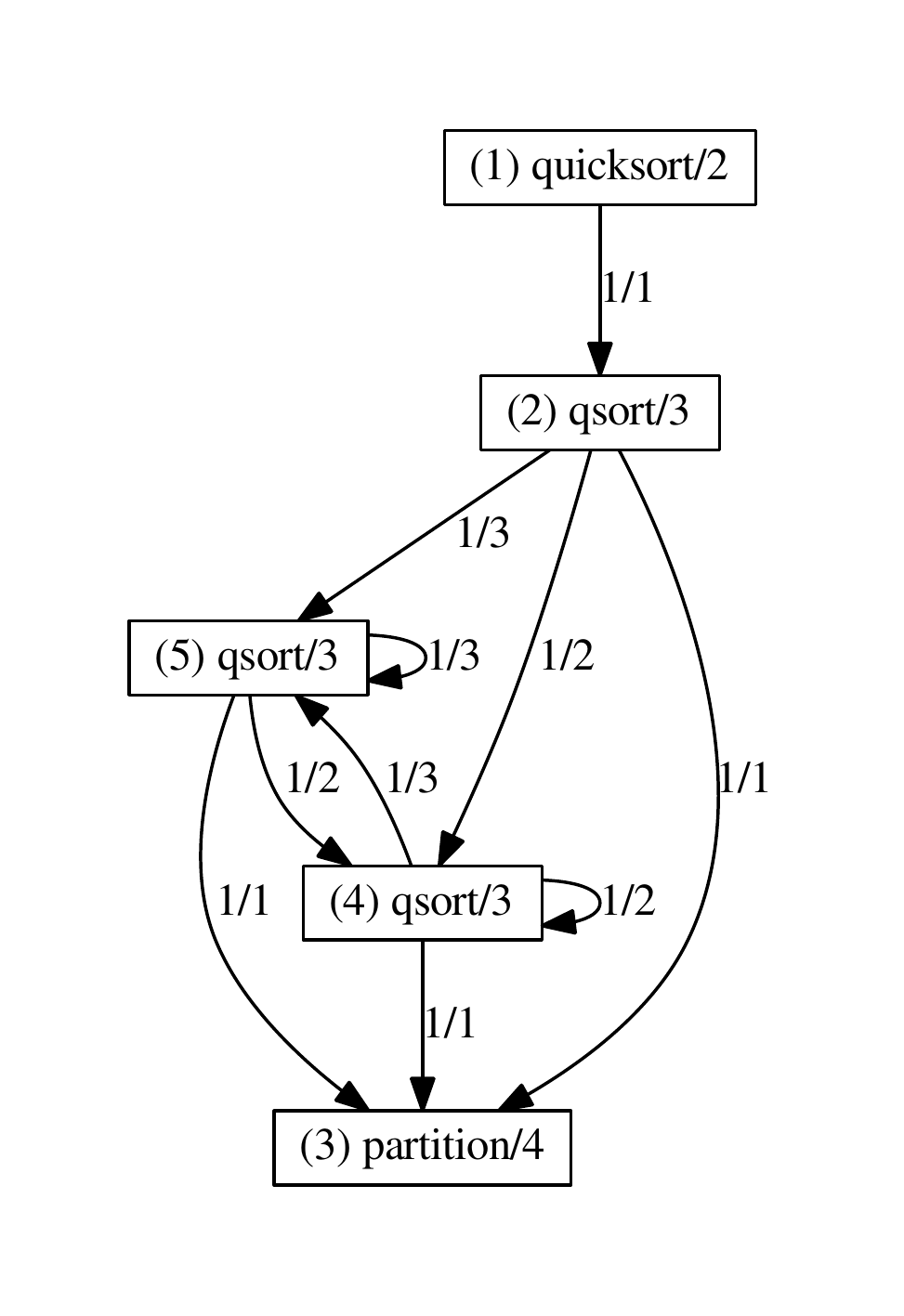}
\end{minipage}}
\ 
\begin{minipage}{0.60\textwidth}
\prettylstciao
\begin{lstlisting}[linewidth=\textwidth,escapechar=@,basicstyle=\scriptsize\ttfamily]
:- module(quicksort,[quicksort/2],[assertions]).
:- use_module(partition,[partition/4]).

:- entry quicksort(Xs,Ys) : (ground(Xs), var(Ys)).

quicksort(Xs,Ys) :-
     qsort(Xs,Ys,[]).

qsort([],Ys,Ys).
qsort([X|Xs],Ys,TailYs) :-
     partition(Xs,X,L,R),
     qsort(R,R2,TailYs),
     qsort(L,Ys,[X|R2]).
\end{lstlisting}
\vfill
\begin{tabular}{ll}
(1) & \protect\scalebox{0.7}{$\langle quicksort(Xs,Ys), ~~ \{Xs/g,Ys/ng\}, ~~ \{Xs/g,Ys/g\} \rangle$}\\
(2) & \protect\scalebox{0.7}{$\langle qsort(Xs,Ys,[]), ~~ \{Xs/g,Ys/ng\}, ~~ \{Xs/g,Ys/g\} \rangle$}\\
(3) & \protect\scalebox{0.7}{$\langle partition(Xs,X,L,R), ~~ \{Xs/g,X/g,L/ng,R/ng\}, ~~ \{Xs/g,X/g,L/g,R/g\} \rangle$}\\
(4) & \protect\scalebox{0.7}{$\langle qsort(Xs,Ys,Zs), ~~ \{Xs/g,Ys/ng,Zs/g\}, ~~ \{Xs/g,Ys/g,Zs/g\} \rangle$}\\
(5) & \protect\scalebox{0.7}{$\langle qsort(Xs,Ys,[Z|Zs]), ~~ \{Xs/g,Ys/ng,Z/g,Zs/g\}, ~~ \{Xs/g,Ys/g,Z/g,Zs/g\} \rangle$}\\
\end{tabular}
\end{minipage}
\caption{Analysis of \texttt{quicksort/2} (using difference lists).}
\label{fig:quicksort-abstree}
\vspace*{-4mm}
\end{figure}
That graph is a finite representation of an infinite abstract and-or
tree. The nodes in the graph correspond to or-nodes
$\langle L, \lambda^c, \lambda^s \rangle$ in the analysis tree, where
the literals $L$, abstract call substitutions $\lambda^c$ and abstract
success substitutions $\lambda^s$ are specified below the graph. The
labels in the edge indicate to which program point each
node corresponds: if one node is connected to its predecessor by an arrow with
label $i/j$, then that node corresponds to the $j$-th literal of the
$i$-th clause of the predicate indicated by the predecessor. The
and-nodes are left implicit.
\end{example-box}

We propose three distances between AND-OR trees $S_1,S_2$ for the same
entry, in increasing order of complexity, and parameterized by a
distance $d_\alpha$ in the underlying abstract domain. We also
discuss which metric properties are inherited by these distances from
$d_\alpha$. Note that a good distance for measuring precision should
fulfill the identity of indiscernibles.

\vspace*{-3mm}
\paragraph{Top distance.} The first consists in considering only the roots of the top trees,
$\langle P, \lambda_c, \lambda_s^1 \rangle_\lor$ and
$\langle P, \lambda_c, \lambda_s^2 \rangle_\lor$, and defining our new
distance as $d(S_1,S_2)=d_\alpha(\lambda_s^1,\lambda_s^2)$. This
distance ignores too much information (e.g., if the entry point is a
predicate \texttt{main/0}, the distance would only distinguish
analyses that detect failure from analysis which do not), so it is not
appropriate for measuring analysis precision, but it is still
interesting as a baseline.
It is straightforward to see that it is a pseudometric if $d_\alpha$
is, but will not fulfill the identity of indiscernibles even if
$d_\alpha$ does.

\vspace*{-3mm}
\paragraph{Flat distance.} The second distance considers all the
information inferred by the analysis for each program point, but
forgetting about its context in the AND-OR tree. In fact, analysis
information is often used this way, i.e., considering only the
substitutions with which a program point can be called or succeeds,
and not which traces lead to those calls (path insensitivity). We
define a distance between program points
$$d_{PP}(S_1,S_2)=\frac{1}{2}(d_\alpha(\bigsqcup_{\lambda \in
  PP_c^1}\lambda,\bigsqcup_{\lambda \in PP_c^2}\lambda) +
d_\alpha(\bigsqcup_{\lambda \in PP_s^1}\lambda,\bigsqcup_{\lambda \in
  PP_s^2}\lambda))$$
where
$PP_c^i=\{ \lambda_c ~|~ \langle PP, \lambda_c, \lambda_s \rangle_\lor
\in S_i \}$,
$PP_s^i=\{ \lambda_s ~|~ \langle PP, \lambda_c, \lambda_s \rangle_\lor
\in S_i \}$. If we denote $P$ as the set of all program points in the
program, that distance can later be extended to a distance between
analyses as $d(S_1,S_2)=\frac{1}{|P|}\sum_{PP \in P}d_{PP}(S_1,S_2)$,
or any other combination of the distances $d_{PP}(S_1,S_2)$ (e.g,
weighted average, $||\cdot||_2$). This distance is more appropriate
for measuring precision than the previous one, but it will still
inherit all metric properties except the identity of indiscernibles.

\vspace*{-3mm}
\paragraph{Tree distance.} For the third distance, we propose the
following recursive definition, which can easily be translated into an
algorithm:

$\\
\begin{array}{l}
d(T_1,T_2)=
 \left\{ \begin{array}{ll}
             \mu\frac{1}{2}(d_\alpha(\lambda_c^1,\lambda_c^2)+d_\alpha(\lambda_s^1,\lambda_s^2))
           + (1-\mu)\frac{1}{|C|}\sum_{(c_1,c_2) \in C}d(c_1,c_2) ~~~~~ if ~ C \neq \emptyset\\
        else ~~~ \frac{1}{2}(d_\alpha(\lambda_c^1,\lambda_c^2)+d_\alpha(\lambda_s^1,\lambda_s^2))
        \end{array} \right. 
\end{array}
\\
$ 

\noindent
  where
$T_1 = \langle P, \lambda_c^1, \lambda_s^1 \rangle, ~ T_2 = \langle
P, \lambda_c^2, \lambda_s^2 \rangle, ~ \mu \in (0,1], ~ C_1 = ch(T_1), ~ C_2 =
ch(T_2)$ and 
$C = \{ (c_1,c_2) ~|~ c_1 \in ch(T_1), c_2 \in ch(T_2), val(c_1)=
\langle X,\_,\_ \rangle, val(c_2) = \langle Y,\_,\_ \rangle, X=Y \}$.
This definition is possible because the two AND-OR trees will
necessarily have the same shape, and therefore we are always comparing
a node with its correspondent node in the other tree.
Also, this distance is well defined, even if the trees, and therefore
the recursions, are infinite, since the expression above always
converges.
Furthermore, the distance to which the expression converges can be
easily computed in finite time. Since the AND-OR trees always have a
finite representation as cyclic trees with $n$ and $m$ nodes
respectively, there are at most $n*m$ different pairs of nodes to
visit during the recursion. Assigning a variable to each pair that is
actually visited, the recursive expression can be expressed as a
linear system of equations. That system has a unique solution since
the original expression had, but also because there is an equation for
each variable and the associated matrix, which is therefore squared,
has strictly dominant diagonal. An example can be found in the
appendix \ref{tree-dist-ex}.

The idea of this distance is that we consider more relevant the
distance between the upper nodes than the distance between the deeper
ones, but we still consider all of them and do not miss any of the
analysis information. As a result, this distance will 
directly inherit the identity of indiscernibles (apart from all other
metric properties) from $d_\alpha$.

\secpre
\section{Experimental Evaluation}
\secpost
\label{experiments}

To evaluate the usefulness of the program analysis distances, we set
up a practical scenario in which we study quantitatively the cost and
precision tradeoff for several abstract domains. In order to do it we
need to overcome two technical problems described below.

\vspace*{-2mm}
\paragraph{Base domain.} Recall that in the distances defined so far,
we assume that we compare two analyses using the same abstract domain.
We relax this requirement by translating each analysis to a common
\emph{base domain}, rich enough to reflect a particular program
property of interest.
An abstract substitution $\lambda$ over a domain $D_\alpha$ is
translated to a new domain $D_{\alpha'}$ as $\lambda'$ =
$\alpha'(\gamma(\lambda))$, and the AND-OR tree is translated by just
translating any abstract substitution occurring in it.
The results still over-approximates concrete executions, but this time
all over the same abstract domain.

\vspace*{-2mm}
\paragraph{Program analysis intersection.} 
Ideally we would compare each analysis with the actual semantics of a
program for a given abstract query, represented also as an AND-OR
tree. However, this semantics is
undecidable
in general,
and we are seeking an automated process. Instead,
we approximated it as the \emph{intersection} of all the computed
analyses.
The intersection between two trees, which can be easily generalized to
$n$ trees, is defined as $inter(T_1,T_2)=T$, with
\vspace*{-4mm}

$$val(T_1) = \langle X, \lambda_c^1, \lambda_s^2 \rangle, ~ val(T_2) =
\langle X, \lambda_c^2, \lambda_s^2 \rangle, ~ val(T) = \langle X,
\lambda_c^1 \sqcap \lambda_c^2, \lambda_s^1 \sqcap \lambda_s^2
\rangle$$

\vspace*{-0.7cm}

$$ch(T) = \{ inter(c_1,c_2) ~|~ c_1 \in ch(T_1),
c_2 \in ch(T_2), val(c_1)= \langle X,\_,\_ \rangle, val(c_2) = \langle
Y,\_,\_ \rangle, X=Y \}$$
That is, a new AND-OR tree with the same shape as those computed by
the analyses, but where each abstract substitution is the greatest
lower bound of the corresponding abstract substitutions in the other
trees.
The resulting tree is the least general AND-OR tree we can obtain that
still over-approximates every concrete execution.

\begin{figure}
  \hspace*{-5mm}
  \includegraphics[width=0.5\textwidth]{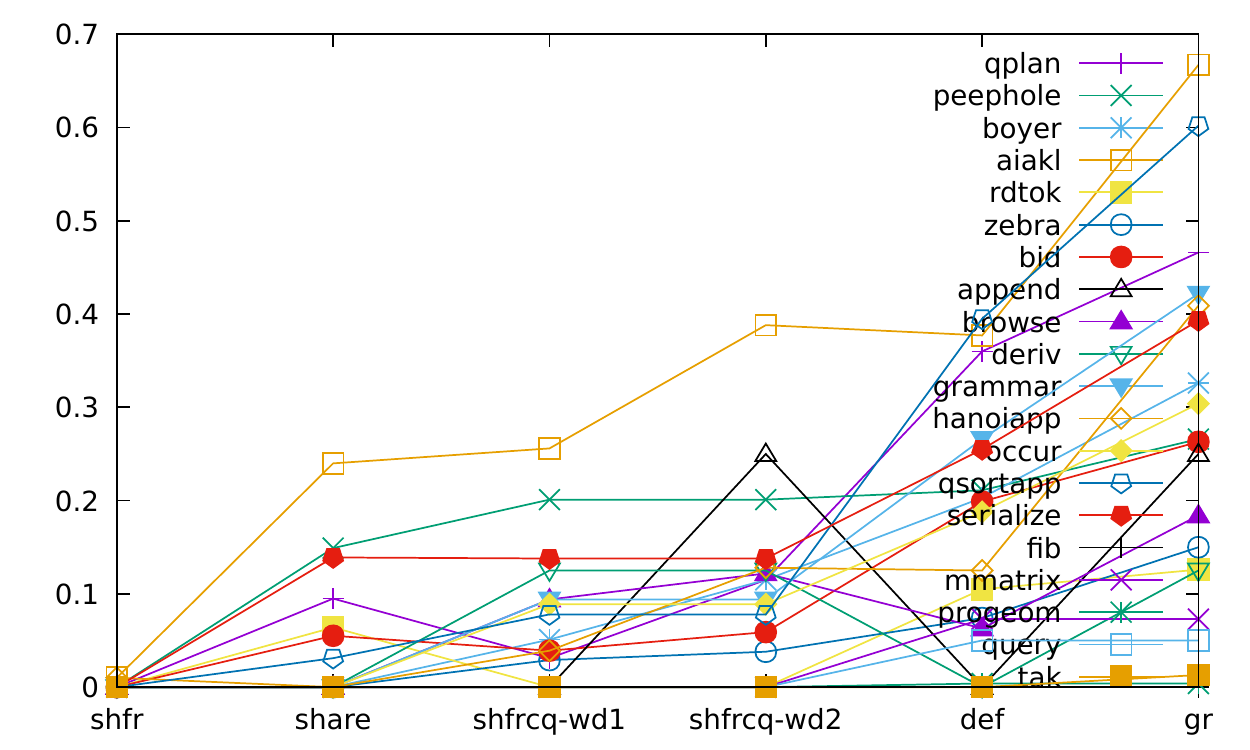}
  \includegraphics[width=0.5\textwidth]{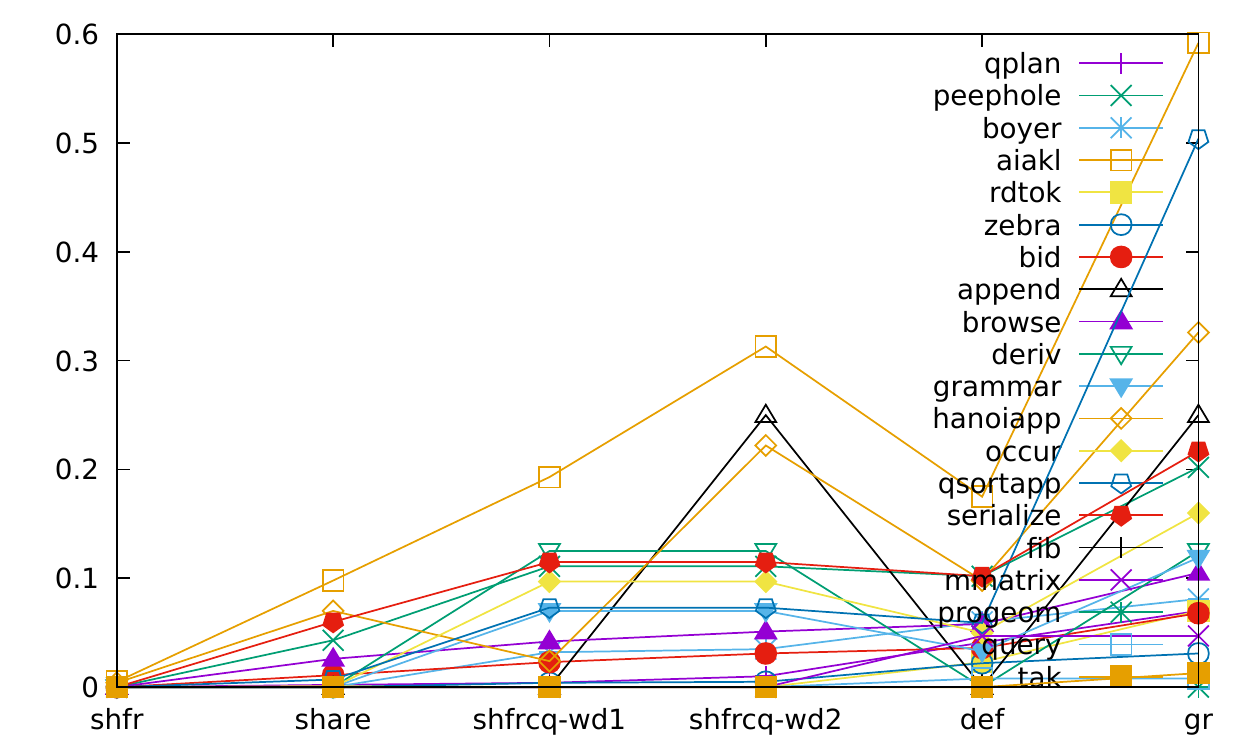}
  \vspace*{-3mm}
  \caption{(a) Precision using flat distance \hfill and \hfill (b) tree
    distance (micro-benchmarks)}
  \label{fig:plot-micro-1}
\vspace*{-2mm}
\end{figure}
\begin{figure}
  \hspace*{-5mm}
  \includegraphics[width=0.5\textwidth]{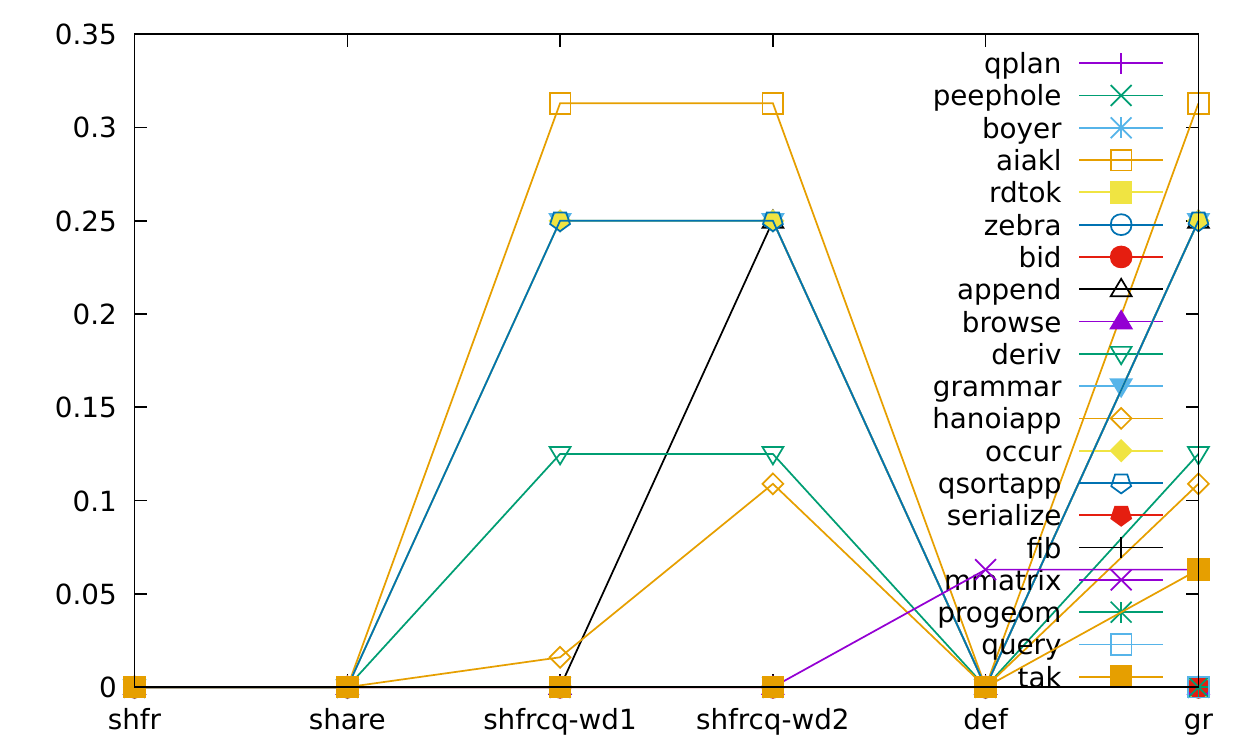}
  \includegraphics[width=0.5\textwidth]{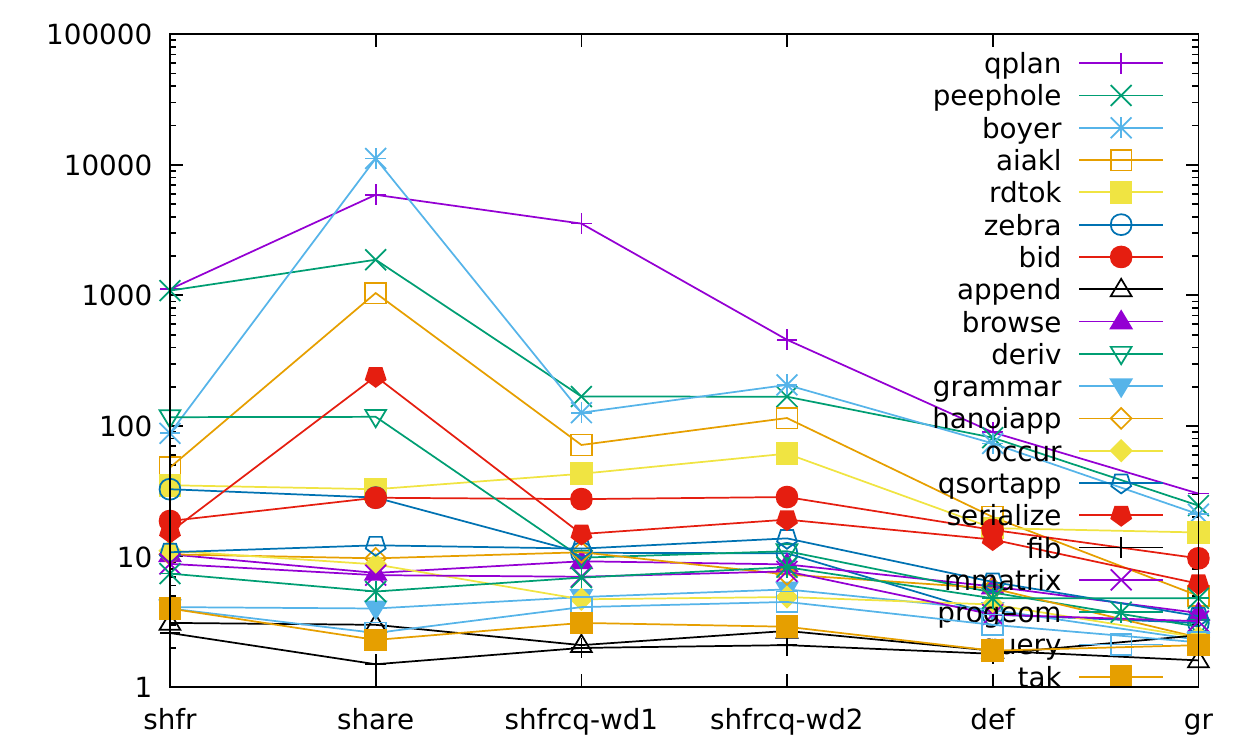}
  \vspace*{-3mm}
  \caption{(a) Precision using top distance  \hfill and \hfill 
    (b) Analysis time (micro-benchmarks)}
  \label{fig:plot-micro-2}
\vspace*{-2mm}
\end{figure}

\begin{figure}
  \hspace*{-5mm}
  \includegraphics[width=0.5\textwidth]{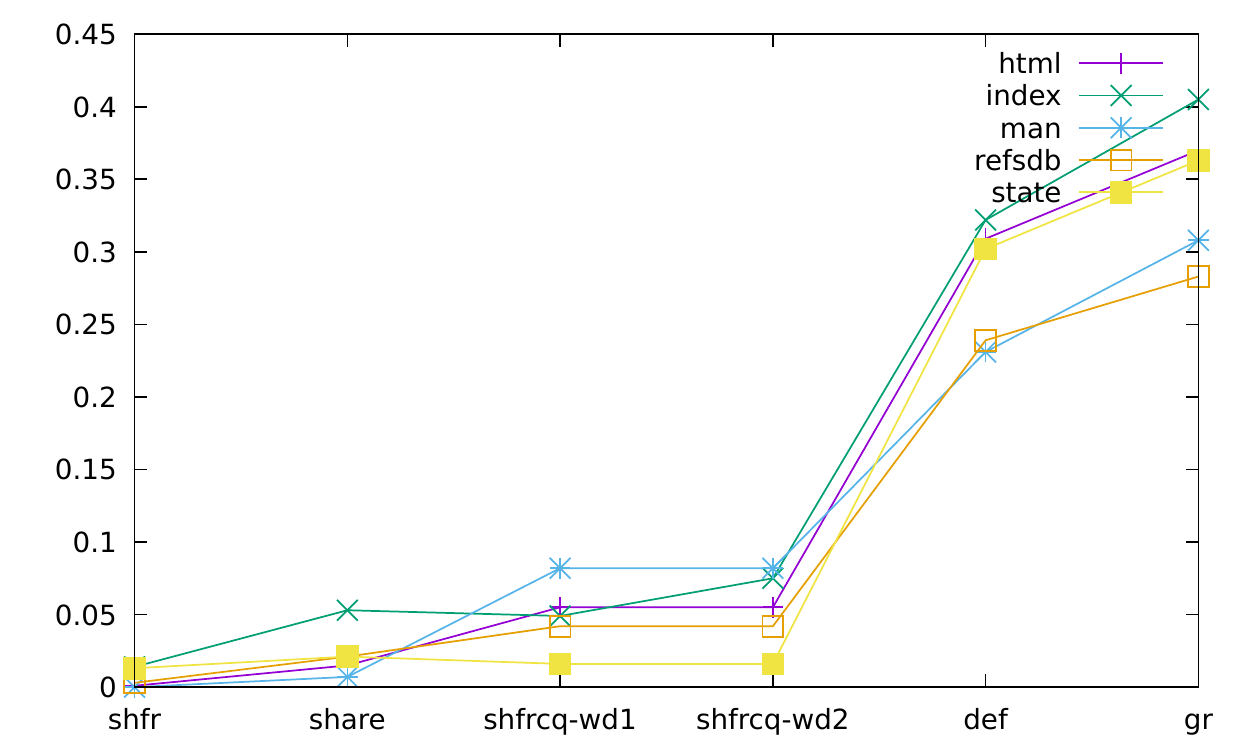}
  \includegraphics[width=0.5\textwidth]{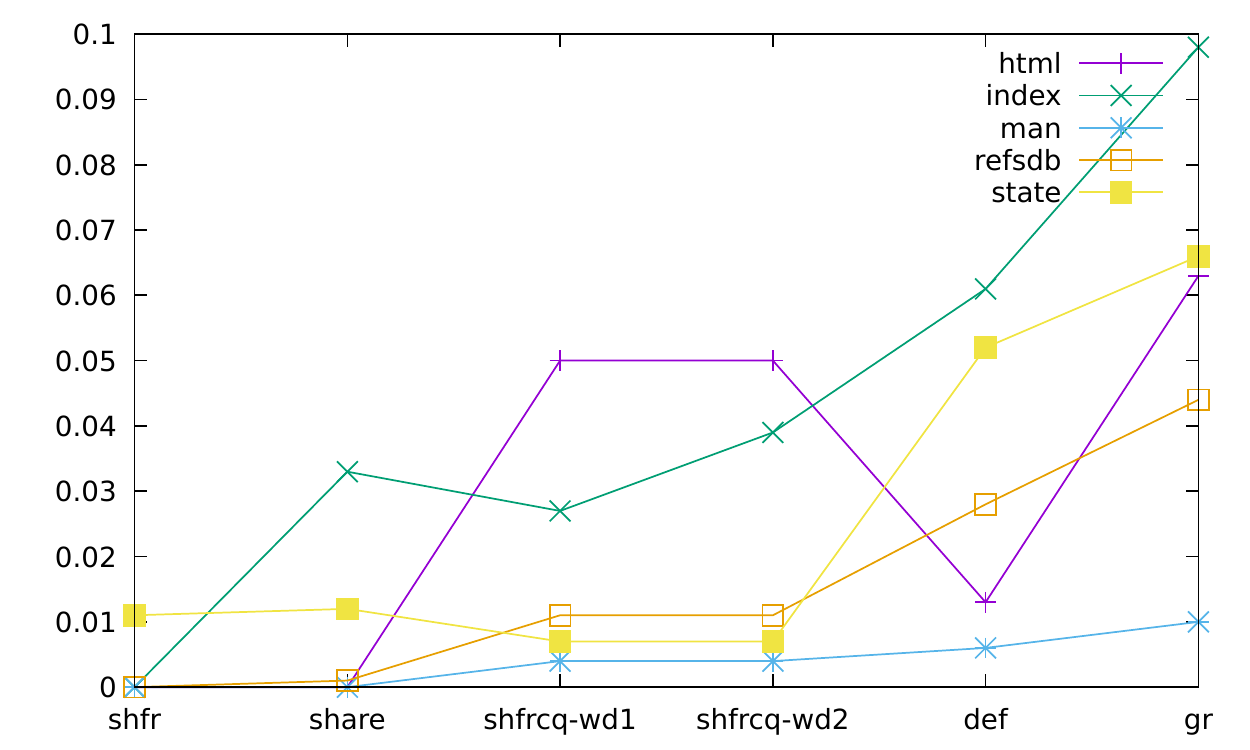}
  \vspace*{-3mm}
  \caption{(a) Precision using flat distance \hfill and \hfill (b) tree
    distance (LPdoc benchmark)}
  \label{fig:plot-lpdoc-1}
\vspace*{-2mm}
\end{figure}
\begin{figure}
  \hspace*{-5mm}
  \includegraphics[width=0.5\textwidth]{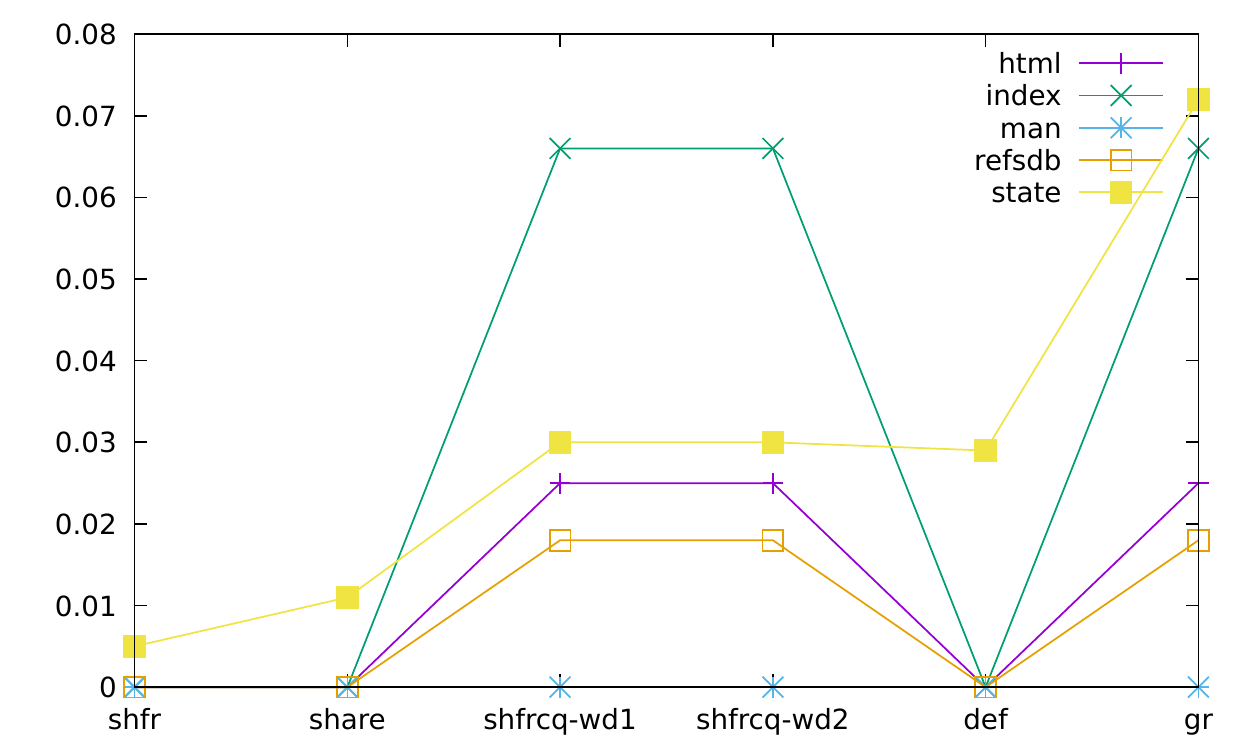}
  \includegraphics[width=0.5\textwidth]{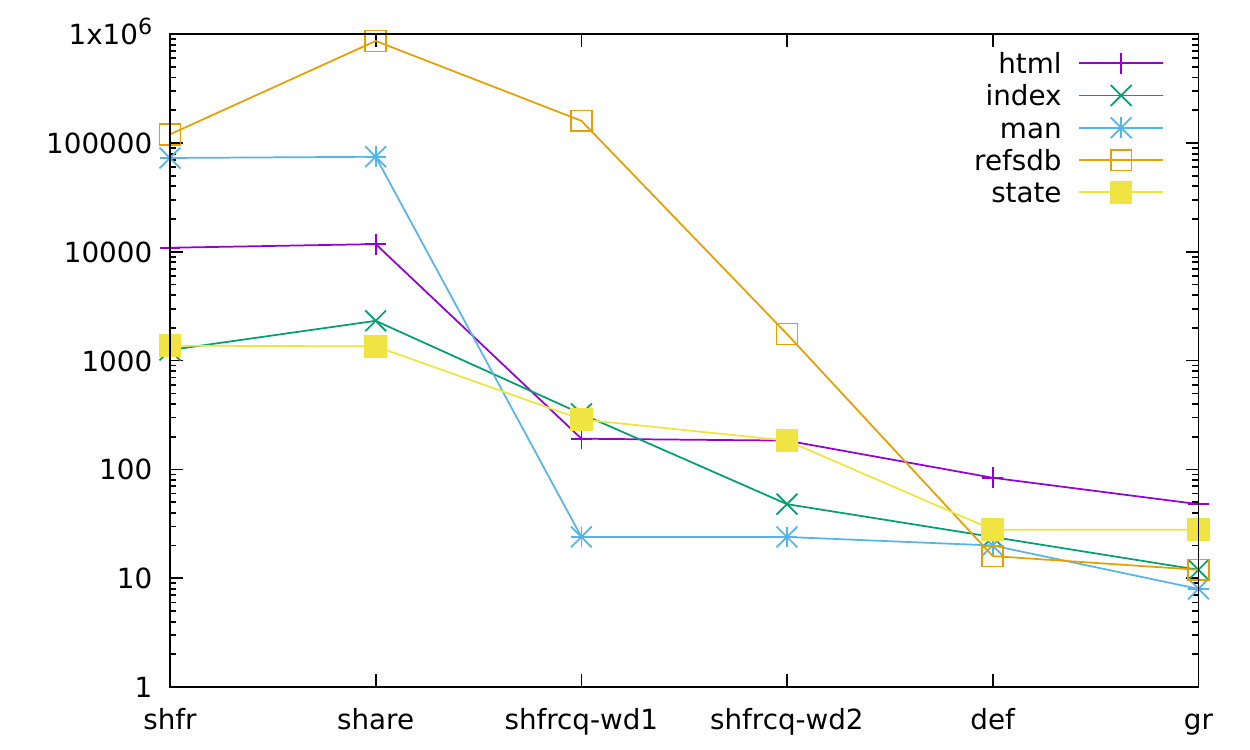}
  \vspace*{-3mm}
  \caption{(a) Precision using top distance  \hfill and \hfill 
    (b) Analysis time (LPdoc benchmarks)}
  \label{fig:plot-lpdoc-2}
\vspace*{-2mm}
\end{figure}

\begin{figure}
  \hspace*{-5mm}
  \includegraphics[width=0.5\textwidth]{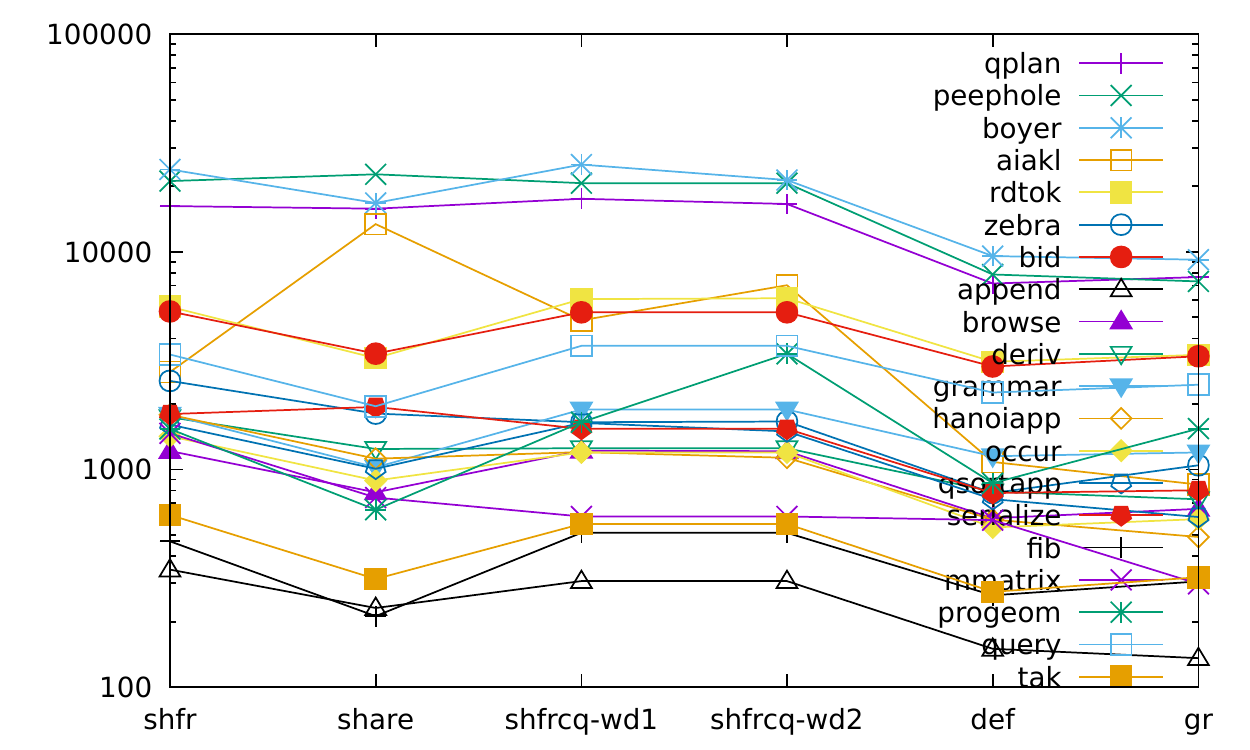}
  \includegraphics[width=0.5\textwidth]{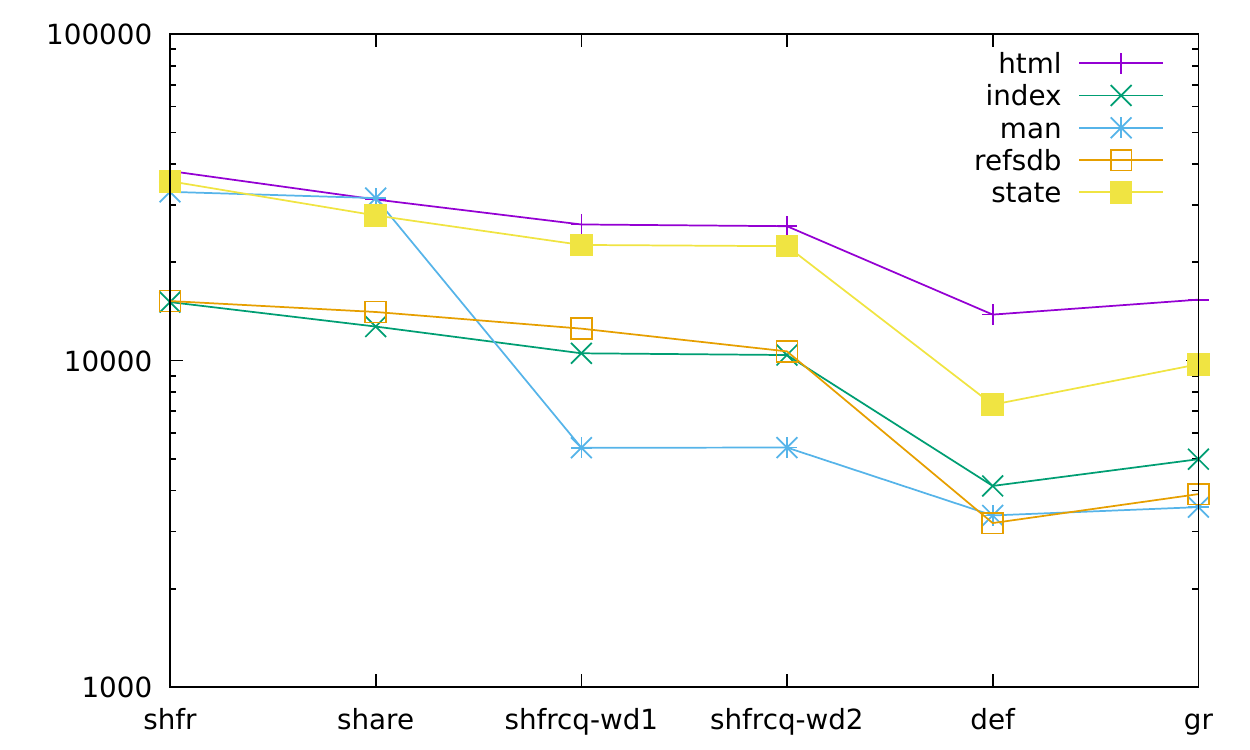}
  \caption{(a) Analysis size (micro-benchmarks) \hfill and \hfill 
    (b) Analysis size (LPdoc benchmark)}
  \label{fig:plot-mem}
  \vspace*{-3mm}
\vspace*{-2mm}
\end{figure}

\vspace*{-2mm}
\paragraph{Case study: variable sharing domains.} 
We have applied the method above on a well known set of
(micro-)benchmarks for CLP analysis, and a number of modules from a real
application (the LPdoc documentation generator).
The programs are analyzed using the CiaoPP framework
\cite{ciaopp-sas03-journal-scp-shortest} and the domains
\textit{shfr}~\cite{freeness-iclp91-short},
\textit{share}~\cite{jacobs89-short,abs-int-naclp89-short},
\textit{def}~\cite{effofree,pos-def94-short}, and
\textit{sharefree\_clique}~\cite{shcliques-padl06-short} with
different widenings.
All these domains express sharing between variables among other
things, and we compare them with respect to the base \textit{share}
domain.
All experiments are run on a Linux machine with Intel Core i5 CPU and
8GB of RAM.

Fig.~\ref{fig:plot-micro-1} and Fig.~\ref{fig:plot-micro-2} show the
results for the micro-benchmarks.
Fig.~\ref{fig:plot-lpdoc-1} and Fig.~\ref{fig:plot-lpdoc-2} show the
same experiment on LPdoc modules.
In both experiments we measure the precision using the flat distance,
tree distance, and top distance.
In general, the results align with our a priori knowledge: that
\textit{shfr} is strictly more precise than all other domains, but
also generally slower; while \textit{gr} is less precise and faster.
As expected, the flat and tree distances show that \textit{share} is
in all cases less precise than \textit{shfr}, and not significantly
cheaper (sometimes even more costly).
The tree distance shows a more pronounced variation of precision when
comparing \textit{share} and widenings. While this can also be
appreciated in the top distance, the top distance fails to show the
difference between \textit{share} and \textit{shfr}. Thus, the tree
distance seems to offer a good balance.
For small programs where analysis requires less than 100ms in
\textit{shfr}, there seems to be no advantage in using less precise
domains.
Also as expected, for large programs widenings provide significant
speedups with moderate precision lose. Small programs do not benefit
in general from widenings.
Finally, the \textit{def} domain shows very good precision w.r.t.\ the
top distance, representing that the domain is good enough to
capture the behavior of predicates at the module interface for the
selected benchmarks.

Fig.~\ref{fig:plot-mem} reflects the size of the AND-OR tree and
experimentally it is correlated with the analysis time. The size
measures of representing abstract substitutions as Prolog terms
(roughly as the number of functor and constant symbols).

\secpre
\section{Related Work}
\secpost
\label{related-work}

\emph{Distances in lattices:}
Lattices and other structures that arise from order relations are
common in many areas of computer science and mathematics, so it is not
surprising that there have been already some attempts at proposing
metrics in them.
E.g.,~\cite{general-lattice-theory} has a dedicated
chapter for metrics in lattices.
\emph{Distances among terms:} Hutch~\cite{hutch-ecml-97},
Nienhuys-Cheng \cite{cheng-hwei-ilp-97} and Jan Ramon
\cite{DBLP:conf/ilp/RamonB98-short} all propose distances in the space of
terms and extend them to distances between sets of terms or
clauses.
Our proposed distance for \textit{regular types} can be interpreted as
the abstraction of the distance proposed by
Nienhuys-Cheng. Furthermore,~\cite{DBLP:conf/ilp/RamonB98-short}
develop some theory of metrics in partial orders, as also does De
Raedt \cite{raedt-09}.
\emph{Distances among abstract elements and operators:}
Logozzo~\cite{logozzo-09} proposes defining metrics in partially
ordered sets and applying them to quantifying the relative loss of
precision induced by numeric abstract domains. Our work is similar in
that we also propose a notion of distance in abstract
domains. However, they restrict their proposed distances to finite or
numeric domains, while we focus instead on logic programming-oriented,
possible infinite, domains. Also, our approach to quantifying the
precision of abstract interpretations follows quite different
ideas.
They use their distances to define a notion of error
induced by an abstract value, and then a notion of error
induced by a finite abstract domain and its abstract operators, with
respect to the concrete domain and concrete operators.
Instead, we work in the context of given programs, and
quantify the difference of precision between the results of different
analyses for those programs, by extending our metrics in abstract
domains to metrics in the space of abstract executions of a program
and comparing those results.
Sotin~\cite{sotin-10} defines measures in $\mathbb{R}^n$ that allow
quantifying the difference in precision between two abstract
values of a numeric domain, by comparing the size of their
concretizations. This is applied to guessing the most appropriate
domain to analyse a program, by under-approximating the potentially
visited states via random testing and comparing the precision with
which different domains would approximate those states.
Di Pierro~\cite{pierro-01} proposes a notion of probabilistic abstract
interpretation, which allows measuring the precision of an abstract
domain and its operators. In their proposed framework, abstract
domains are vector spaces instead of partially ordered sets, and it is
not clear whether every domain, and in particular those used in logic
programming, can be reinterpreted within that framework.
Cortesi~\cite{Cortfilwin92} proposes a formal methodology to compare
qualitatively the precision of two abstract domains with respect to
some of the information they express, that is, to know if one is
strictly more precise that the other according to only part of the
properties they abstract. In our experiments, we compare the precision
of different analyses with respect to some of the information they
express. For some, we know that one is qualitatively more precise than
the other in Cortesi's paper's sense, and that is reflected in our
results.

\secpre
\section{Conclusions}
\secpost

We have proposed a new approach for measuring and comparing precision
across different analyses, based on defining distances in abstract
domains and extending them to distances between whole analyses.
We have surveyed and extended previous proposals for distances and
metrics in lattices or abstract domains, and
proposed metrics for some common (C)LP domains.
We have also proposed extensions of those metrics to the
space of whole program analysis. We have implemented those metrics
and applied them to measuring the
precision of different sharing-related (C)LP analyses on both
benchmarks and a realistic program.
We believe that this application of distances is promising for
debugging the precision of analyses and calibrating heuristics for
combining different domains in portfolio approaches, without prior
knowledge and treating domains as black boxes (except for the
translation to the \emph{base} domain).
In the future we plan to apply the proposed concepts in other
applications beyond measuring precision in analysis, such as studying
how programming methodologies or optimizations affect the analyses,
comparing obfuscated programs, giving approximate results in semantic
code browsing~\cite{deepfind-iclp2016-short}, program synthesis,
software metrics, etc.

\vspace*{4mm}
\noindent
\textbf{Acknowledgements:} Research partially funded by EU FP7
agreement no 318337 \emph{ENTRA}, Spanish MINECO TIN2015-67522-C3-1-R
\emph{TRACES} project, the Madrid M141047003 \emph{N-GREENS} program
and \emph{BLOQUES-CM} project, and the TEZOS Foundation \emph{TEZOS}
project.
  

\begin{small}

\end{small}

\clearpage
\appendix
\section{Theory of Section \ref{domain-distances}}

\subsection{Properties inherited by abstraction or concretization of distances}

\begin{proposition}

  Let us consider an abstract domain $D_\alpha$, that abstracts the
  concrete domain $D$, with abstraction function
  $\alpha : D \rightarrow D_\alpha$ and concretization function
  $\gamma : D_\alpha \rightarrow D$. Both domains are complete
  lattices and $\alpha$ and $\gamma$ form a Galois connection. Then:

(1) If $d_\alpha : D_\alpha \times D_\alpha \rightarrow \mathbb{R}$ is
  a metric in the abstract domain, then $d : D \times D \rightarrow
  \mathbb{R}, ~ d(A,B) = d_\alpha(\alpha(A),\alpha(B))$ is a
  pseudometric in the concrete domain. If $d_\alpha$ is
  \textit{order-preserving}, so it is $d$.

(2) If $d : D \times D \rightarrow \mathbb{R}$ is a metric in the
  concrete domain, then $d_\alpha : D_\alpha \times D_\alpha
  \rightarrow \mathbb{R}, ~ d_\alpha(a,b) = d(\gamma(a),\gamma(b))$ is
  a pseudometric in the abstract domain. If the Galois connection is a
  Galois insertion, then $d$ is a full metric. If $d$ is
  \textit{order-preserving}, so it is $d_\alpha$.

\end{proposition}

\begin{proof}[\textbf{Proof}]
\label{proof:conc-dist-abs-dist}

\begin{itemize}

\item (1)
  
  \begin{itemize}
      
  \item $d$ is a pseudometric:

    \begin{itemize}

      \setlength\itemsep{-1.5em}

    \item Non-negativity: $d(A,B)=d_\alpha(\alpha(A),\alpha(B)) \geq
      0$, since $d_\alpha$ is non-negative \\
    \item Weak identity of indiscernibles :
      $d(A,A)=d_\alpha(\alpha(A),\alpha(A))=0$, since $d_\alpha$
      fulfills the identity of indiscernibles \\
    \item Symmetry: $d(A,B)=d_\alpha(\alpha(A),\alpha(B)) =
      d_\alpha(\alpha(B),\alpha(A)) = d(B,A)$, since $d_\alpha$ is
      symmetric \\
    \item Triangle inequality: $d(A,C)=d_\alpha(\alpha(A),\alpha(C))
      \leq d_\alpha(\alpha(A),\alpha(B)) +
      d_\alpha(\alpha(B),\alpha(C)) = d(A,B) + d(B,C)$, since
      $d_\alpha$ fulfills the triangle inequality \\
      \end{itemize}
      
  \item $d$ is \textit{order-preserving}:
      
    If $A \subseteq B \subseteq C$, then $\alpha(A) \sqsubseteq
    \alpha(B) \sqsubseteq \alpha(C)$, since $\alpha$ is monotonic. But
    then $d(A,B)=d_\alpha(\alpha(A),\alpha(B)) \leq
    d_\alpha(\alpha(A),\alpha(C)) = d(A,C)$, since $d_\alpha$ is
    order-preserving.
      
  \end{itemize}
    
\item (2)

  \begin{itemize}

  \item $d_\alpha$ is a pseudometric: analogous. Besides, if the
    Galois connection is a Galois insertion, then $\gamma$ is
    injective (otherwise, $\exists ~ a \neq b \in D_\alpha ~ s.t. ~
    \gamma(a)=\gamma(b) \implies \alpha(\gamma(a))=\alpha(\gamma(b))
    \implies a=b$, which is absurd). But then $d_\alpha(a,b)=0
    \implies d(\gamma(a),\gamma(b))=0 \implies \gamma(a)=\gamma(b)
    \implies a=b$, and therefore $d_\alpha$ is a full metric
    
  \item $d_\alpha$ is \textit{order-preserving}: Analogous

  \end{itemize}

\end{itemize}

\end{proof}

\section{Examples for section 4}

\subsection{Example of \textit{program-points} distance}

  The analysis shown in Fig. \ref{fig:quicksort-abstree} has only one
  triple $\langle L, \lambda^c, \lambda^s \rangle$ for each program
  point. Let us consider a different analysis for the same program, in
  which there is no information about the imported predicate
  \texttt{partition/4}, and therefore the analysis needs to assume the
  most general abstract substitution on success for calls to that
  predicate. Fig.
  \ref{fig:quicksort-abstree-2} shows the result of the analysis in
  the same manner as Fig. \ref{fig:quicksort-abstree} does. We observe
  that this time there are program points which have more that one
  triple in the analysis. Let us denote each program point as
  \texttt{P/A/N/M}, where that represents the \texttt{M}-th literal of
  the \texttt{N}-th clause of the predicate \texttt{P/A}. The
  correspondence between program points and analysis nodes is the
  following: \\

\begin{figure}[t]
  \centering
  \includegraphics[scale=0.6]{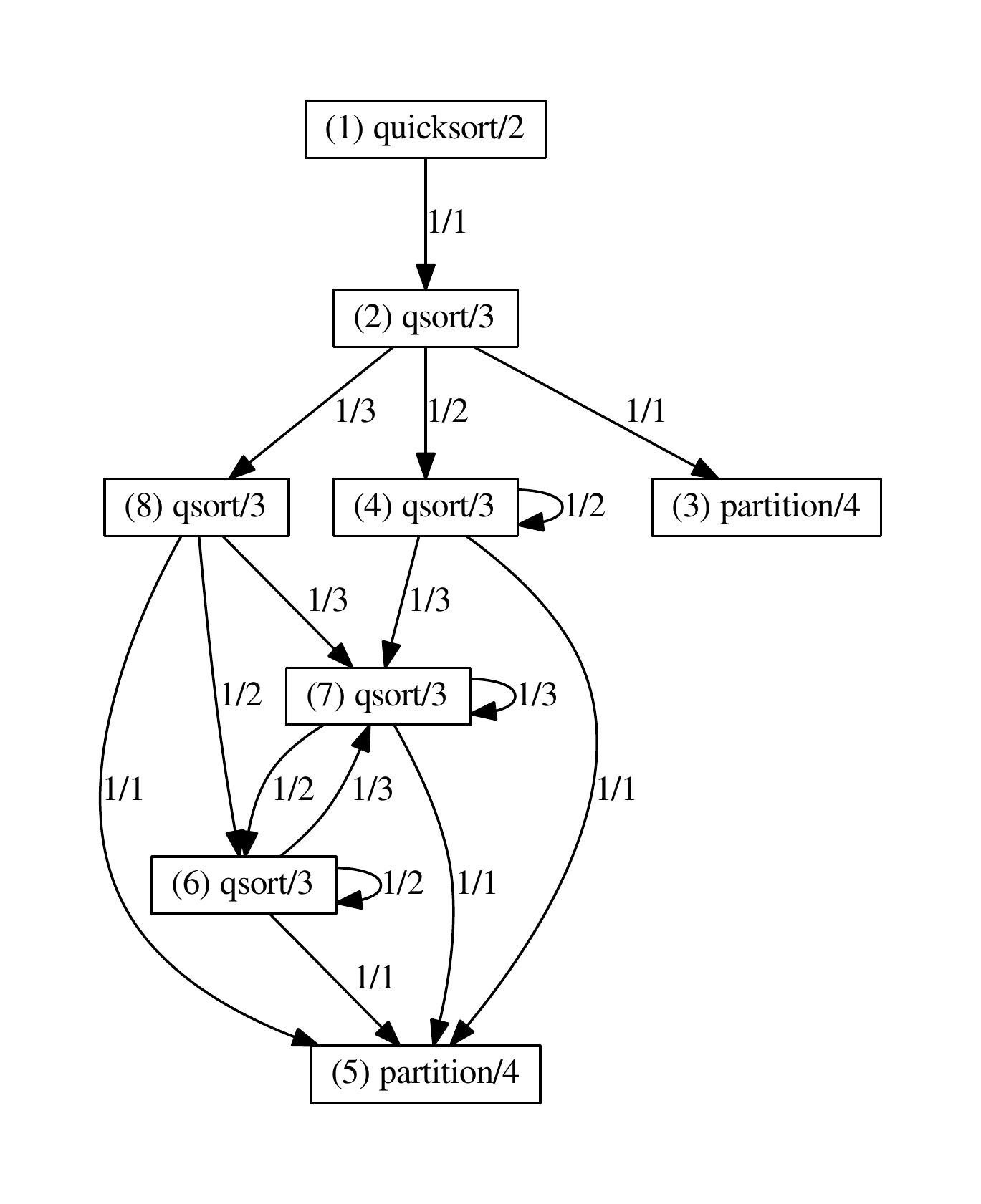}
\begin{tabular}{ll}
(1)  & \protect\scalebox{0.7}{$\langle quicksort(Xs,Ys), ~~ \{Xs/g,Ys/ng\}, ~~ \{Xs/g,Ys/any\} \rangle$}\\
(2)  & \protect\scalebox{0.7}{$\langle qsort(Xs,Ys,[]), ~~ \{Xs/g,Ys/ng\}, ~~ \{Xs/g,Ys/any\} \rangle$}\\
(3)  & \protect\scalebox{0.7}{$\langle partition(Xs,X,L,R), ~~ \{Xs/g,X/g,L/ng,R/ng\}), ~~ \{Xs/g,X/g,L/any,R/any\}) \rangle$}\\
(4)  & \protect\scalebox{0.7}{$\langle qsort(Xs,Ys,Zs), ~~ \{Xs/any,Ys/ng,Zs/g\}), ~~ \{Xs/any,Ys/any,Zs/g\}) \rangle$}\\
(5)  & \protect\scalebox{0.7}{$\langle partition(Xs,X,L,R), ~~ \{Xs/any,X/any,L/ng,R/ng\}), ~~ \{Xs/any,X/any,L/any,R/any\}) \rangle$}\\
(6)  & \protect\scalebox{0.7}{$\langle qsort(Xs,Ys,Zs), ~~ \{Xs/any,Ys/ng,Zs/any\}, ~~ \{Xs/any,Ys/any,Zs/any\} \rangle$}\\
(7)  & \protect\scalebox{0.7}{$\langle qsort(Xs,Ys,[Z|Zs]), ~~ \{Xs/any,Ys/ng,Z/any,Zs/any\}), ~~ \{Xs/any,Ys/any,Z/any,Zs/any\}) \rangle$}\\
(8)  & \protect\scalebox{0.7}{$\langle qsort(Xs,Ys,[Z|Zs]), ~~ \{Xs/any,Ys/ng,Z/g,Zs/any\}, ~~ \{Xs/any,Ys/any,Z/g,Zs/any\} \rangle$}\\
\end{tabular}
\caption{Analysis of \texttt{quicksort/2}.}
  \label{fig:quicksort-abstree-2}
\end{figure}
  
\begin{small}
\begin{tabular}{|c|c|c|c|c|}
\cline{1-5}
\texttt{quicksort/2/0} (entry) & \texttt{quicksort/2/1/1} & \texttt{qsort/3/1/1} & \texttt{qsort/3/1/2} & \texttt{qsort/3/1/3} \\ \cline{1-5}
(1) & (2) & (3), (5) & (4), (6) & (7), (8) \\ \cline{1-5}
\end{tabular}
\end{small}

The resulting single triples $\langle L, \lambda^c, \lambda^s \rangle$
for each program point will be the following: \\

\begin{footnotesize}
\begin{tabular}{|c|c|c|}
\cline{1-3}
\texttt{quicksort/2/0} (entry) & (1) & \protect\scalebox{0.7}{$\langle quicksort(Xs,Ys), ~~ \{Xs/g,Ys/ng\}, ~~ \{Xs/g,Ys/any\} \rangle$} \\ \cline{1-3}
\texttt{quicksort/2/1/1} & (2) & \protect\scalebox{0.7}{$\langle qsort(Xs,Ys,[]), ~~ \{Xs/g,Ys/ng\}, ~~ \{Xs/g,Ys/any\} \rangle$} \\ \cline{1-3}
\texttt{qsort/3/1/1} & (3) '$\sqcup$' (5) & \protect\scalebox{0.7}{$\langle partition(Xs,X,L,R), ~~ \{Xs/any,X/any,L/ng,R/ng\}), ~~ \{Xs/any,X/any,L/any,R/any\}) \rangle$} \\ \cline{1-3}
\texttt{qsort/3/1/2} & (4) '$\sqcup$' (6) & \protect\scalebox{0.7}{$\langle qsort(Xs,Ys,Zs), ~~ \{Xs/any,Ys/ng,Zs/any\}, ~~ \{Xs/any,Ys/any,Zs/any\} \rangle$}\\ \cline{1-3}
\texttt{qsort/3/1/3} & (7) '$\sqcup$' (8) & \protect\scalebox{0.7}{$\langle qsort(Xs,Ys,[Z|Zs]), ~~ \{Xs/any,Ys/ng,Z/any,Zs/any\}, ~~ \{Xs/any,Ys/any,Z/g,Zs/any\} \rangle$} \\ \cline{1-3}
\end{tabular}
\end{footnotesize}

Let us compare the two analyses shown in Figs.
\ref{fig:quicksort-abstree} and \ref{fig:quicksort-abstree-2}. We
already have their representation as one triple
$\langle L, \lambda^c, \lambda^s \rangle$ for each program point. The
distances for each program point, computed as the average of the
distance between its abstract call substitution and the distance
between its abstract success substitution, is the
following: \\
\begin{small}
  \begin{tabular}{|c|c|c|c|c|}
    \cline{1-5}
    \texttt{quicksort/2/0} (entry) & \texttt{quicksort/2/1/1} & \texttt{qsort/3/1/1} & \texttt{qsort/3/1/2} & \texttt{qsort/3/1/3} \\ \cline{1-5}
    0.354 & 0.354 & 0.427 & 0.454 & 0.467 \\ \cline{1-5}
  \end{tabular}
\end{small}

The final distance between the analysis could be the average of all of
them, 0.411. Alternatively, we could assign different weights to each
program point taking into account the structure of the program, and
use a weighted average as final distance. For example, we could assign
the weights of the table below, which would yield the final distance
0.378. \\

\begin{small}
  \begin{tabular}{|c|c|c|c|c|}
    \cline{1-5}
    \texttt{quicksort/2/0} (entry) & \texttt{quicksort/2/1/1} & \texttt{qsort/3/1/1} &  \texttt{qsort/3/1/2} & \texttt{qsort/3/1/3} \\ \cline{1-5} 
    $\frac{1}{2}$ & $\frac{1}{4}$ & $\frac{1}{12}$ & $\frac{1}{12}$ & $\frac{1}{12}$ \\ [1mm] \cline{1-5}
  \end{tabular}
\end{small}

\subsection{Example of the \textit{tree} distance}
\label{tree-dist-ex}

Let us compute the \textit{tree} distance between the two analyses
shown in Figs. \ref{fig:quicksort-abstree} and
\ref{fig:quicksort-abstree-2}. Fig. \ref{fig:qsort-3rd-approach} shows
the tree with distances between both analysis node to node. The
and-nodes are omitted for simplicity. Each or-node is a quintuple
$(P,Id_1,Id_2,D,W)$: $P$ is the predicate corresponding to that
program point, $I_1$ is the identifier of the node in analysis
\ref{fig:quicksort-abstree} corresponding to that or-node, $I_2$ is
the analogous in analysis \ref{fig:quicksort-abstree-2}, $D$ is the
distance between the two nodes, and $W$ is the corresponding weight to
the distance in that node when we apply the definition of the
\textit{tree} distance. We use a factor $\mu=\frac{1}{5}$, and the
average of the distance between the call substitutions and the
distance between the success substitutions as distance between nodes,
using an abstract distance in the underlying \textit{groundness}
domain.

\begin{figure}[t]
  \hspace*{-2cm}
  \includegraphics[scale=0.5]{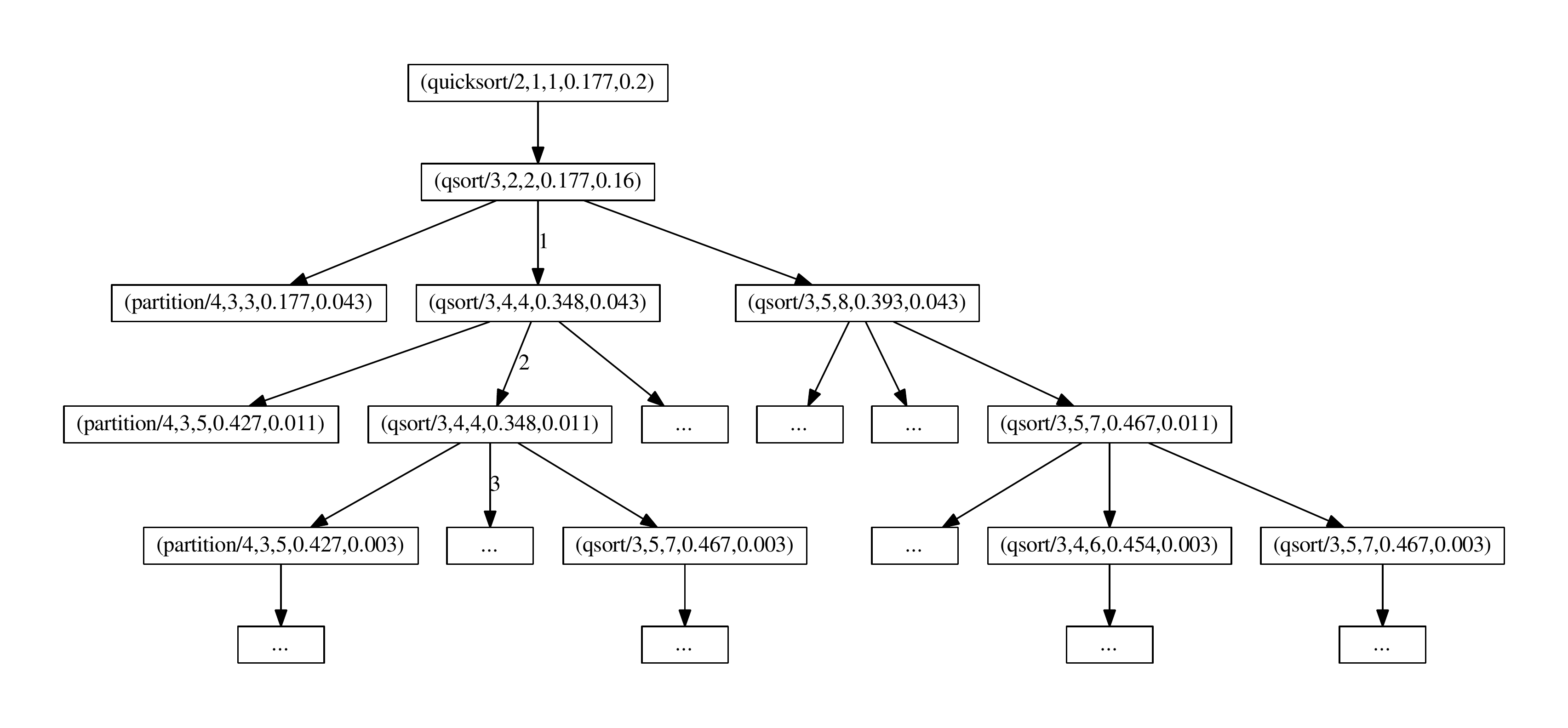}
  \caption{3rd approach: whole abstract execution tree}
  \label{fig:qsort-3rd-approach}
\end{figure}

If we follow the tree through the edges labelled 1,2,3..., we observe
that we are visiting the same node over and over with decreasing
weights
$0.043, 0.011, 0.003 \ldots = w \frac{1}{5} + w \frac{4}{5}
\frac{1}{3} \frac{1}{5} + w \frac{4}{5} \frac{1}{3} \frac{4}{5}
\frac{1}{3} \frac{1}{5} + \ldots$, where
$w=1 \frac {4}{5} \frac{1}{1} \frac{4}{5} \frac{1}{3}$. The sum of
those weights converges
($\frac{1}{5} w \sum_{i=0}^\infty{(\frac{4}{5} \frac{1}{3})^i} =
\frac{1}{5} w \frac{15}{11}$), but it is not trivial to compute in the
general case and for all cases.

However, we can compute the final sum solving the following systems of
equations, where the variable $X_{i,j}$ corresponds to the node $(P,i,j,D,W)$:

$\\
\begin{array}{l}

 \left\{ \begin{array}{ll}
           X_{1,1} = \frac{1}{5}*0.177 + \frac{4}{5}X_{2,2}\\
           X_{2,2} = \frac{1}{5}*0.177 + \frac{4}{5}\frac{1}{3}X_{3,3} + \frac{4}{5}\frac{1}{3}X_{4,4} + \frac{4}{5}\frac{1}{3}X_{5,8} \\
           X_{3,3} = 0.177 \\
           X_{4,4} = \frac{1}{5}*0.348 + \frac{4}{5}\frac{1}{3}X_{3,5} + \frac{4}{5}\frac{1}{3}X_{4,4} + \frac{4}{5}\frac{1}{3}X_{5,7} \\
           X_{5,8} = \frac{1}{5}*0.177 + \frac{4}{5}\frac{1}{3}X_{3,5} + \frac{4}{5}\frac{1}{3}X_{4,6} + \frac{4}{5}\frac{1}{3}X_{5,7} \\
           X_{3,5} = 0.427 \\
           X_{5,7} = \frac{1}{5}*0.177 + \frac{4}{5}\frac{1}{3}X_{3,5} + \frac{4}{5}\frac{1}{3}X_{4,6} + \frac{4}{5}\frac{1}{3}X_{5,7} \\
           X_{4,6} = \frac{1}{5}*0.177 + \frac{4}{5}\frac{1}{3}X_{3,5} + \frac{4}{5}\frac{1}{3}X_{4,6} + \frac{4}{5}\frac{1}{3}X_{5,7} \\
 \end{array} \right.
               
\end{array}
\\
$ 

\end{document}